\documentclass[preprint,superscriptaddress,
showpacs,preprintnumbers,amsmath,amssymb]{revtex4}

\usepackage{amssymb,amsmath,amsthm}
\usepackage{epsfig}
\usepackage{multirow}
\newtheorem{Thm}{Theorem}

\newtheorem{Lem}[Thm]{Lemma}

\theoremstyle{definition}

\newcommand{\bra}[1]{{\left\langle #1 \right|}}
\newcommand{\ket}[1]{{\left| #1 \right\rangle}}

\newcommand{\T}{\mbox{$\mathrm{tr}$}}

\begin{document}
\title{Hamming weight and tight constraints of multi-qubit entanglement in terms of unified entropy}

\author{Jeong San Kim}
\email{freddie1@khu.ac.kr} \affiliation{
 Department of Applied Mathematics and Institute of Natural Sciences, Kyung Hee University, Yongin-si, Gyeonggi-do 446-701, Korea
}
\date{\today}

\begin{abstract}
We establish a characterization of multi-qubit entanglement constraints in terms of non-negative power of entanglement measures based on unified-$(q,s)$ entropy. Using the Hamming weight of the binary vector related with the distribution of subsystems, we establish a class of tight monogamy inequalities of multi-qubit entanglement based on the $\alpha$th-power of unified-$(q,s)$ entanglement for $\alpha \geq 1$. For $0 \leq \beta \leq 1$, we establish a class of tight polygamy inequalities of multi-qubit entanglement in terms
of the $\beta$th-power of unified-$(q,s)$ entanglement of assistance. Thus our results characterize the monogamy and polygamy of multi-qubit entanglement for the full range of non-negative power of unified entanglement.
\end{abstract}

\maketitle

\section{Introduction}
\label{Intro}
Quantum entanglement is a quintessential feature of quantum mechanics revealing the fundamental insights into the nature of quantum correlations.
One distinct property of quantum entanglement without any classical counterpart is its limited shareability in multi-party quantum systems, known as the {\em monogamy of entanglement}(MoE)~\cite{T04, KGS}. MoE is the fundamental ingredient for secure quantum cryptography~\cite{qkd1, qkd2},
and it also plays an important role in condensed-matter physics such as the $N$-representability problem for
fermions~\cite{anti}.

Mathematically, MoE is characterized in a quantitative way known as the {\em monogamy inequality}; for a three-qubit quantum state $\rho_{ABC}$
with its two-qubit reduced density matrices $\rho_{AB}=\T_C \rho_{ABC}$ and
$\rho_{AC}=\T_B \rho_{ABC}$, the first monogamy inequality was established by Coffman-Kundu-Wootters(CKW) as
\begin{equation*}
\tau\left(\rho_{A|BC}\right)\geq \tau\left(\rho_{A|B}\right)+\tau\left(\rho_{A|C}\right)
\label{MoE}
\end{equation*}
where $\tau\left(\rho_{A|BC}\right)$ is the bipartite entanglement between subsystems $A$ and $BC$, quantified by {\em tangle} and
$\tau\left(\rho_{A|B}\right)$ and $\tau\left(\rho_{A|C}\right)$ are the tangle between $A$ and $B$
and between $A$ and $C$, respectively~\cite{CKW}.

CKW inequality shows the mutually exclusive relation of two-qubit entanglement
between $A$ and each of $B$ and $C$ measured by $\tau\left(\rho_{A|B}\right)$ and $\tau\left(\rho_{A|C}\right)$ respectively,
so that their summation cannot exceeds the total entanglement between $A$ and $BC$, that is, $\tau\left(\rho_{A|BC}\right)$.
Later, three-qubit CKW inequality was generalized for arbitrary multi-qubit systems~\cite{OV} and
some cases of multi-party, higher-dimensional quantum systems more than qubits in terms of various
bipartite entanglement measures~\cite{KDS, KSRenyi, KimT, KSU}.

Whereas entanglement monogamy characterizes the limited shareability of entanglement in multi-party quantum systems,
the {\em assisted entanglement}, which is a dual amount to bipartite entanglement measures, is also known to be
dually monogamous, thus polygamous in multi-party quantum systems;
for a three-qubit state $\rho_{ABC}$, a {\em polygamy inequality} was proposed as
\begin{equation*}
\tau^a\left(\rho_{A|BC}\right)\leq\tau^a\left(\rho_{A|B}\right)
+\tau^a\left(\rho_{A|C}\right),
\label{PoE}
\end{equation*}
where $\tau^a\left(\rho_{A|BC}\right)$ is the tangle of assistance~\cite{GMS, GBS}.
Later, the tangle-based polygamy inequality of entanglement was generalized into multi-qubit systems as well as some class of
higher-dimensional quantum systems using various entropic entanglement measures~\cite{KimT, BGK, KUP}. General polygamy inequalities of entanglement were also established in arbitrary dimensional multi-party quantum systems~\cite{KimGP, KimGP16}.

Recently, a new class of monogamy inequalities using the $\alpha$th-power
of entanglement measures were proposed;
it was shown that the $\alpha$th-power of {\em entanglement of formation} and concurrence can be used to establish multi-qubit monogamy inequalities for $\alpha \geq \sqrt{2}$ and $\alpha \geq 2$, respectively~\cite{Fei}.
Later, tighter monogamy and polygamy inequalities of entanglement using non-negative power of concurrence and {\em squar of convex-roof extended negativity} were also proposed for multi-qubit systems~\cite{Fei2, Kim18}.

Here, we provide a full characterization of multi-qubit entanglement monogamy and polygamy constraints in terms of non-negative power of entanglement measures based on unified entropy~\cite{ue1, ue2}.
Using the Hamming weight of the binary vector related with the distribution of subsystems, we establish a class of tight monogamy inequalities of multi-qubit entanglement based on the $\alpha$th-power of unified-$(q,s)$ entanglement~\cite{KSU} for $\alpha \geq 1$.
For $0 \leq \beta \leq 1$, we establish a class of tight polygamy inequalities of multi-qubit entanglement in terms
of the $\beta$th-power of unified-$(q,s)$ entanglement of assistance~\cite{KUP}.

This paper is organized as follows. In Section~\ref{Sec: unimeasure}, we review the definitions of unified entropy, unified-$(q,s)$ entanglement and unified-$(q,s)$ entanglement of assistance as well as multi-qubit monogamy and polygamy inequalities in terms of unified entanglements. In Section~\ref{Sec: TM_qubit}, we establish a class of tight monogamy inequalities in multi-qubit system based on the $\alpha$th-power of
unified-$(q,s)$ entanglement for $\alpha \geq 1$.
In Section~\ref{Sec: TP_qubit}, we establish a class of tight polygamy inequalities of multi-qubit entanglement in terms
of the $\beta$th-power of unified-$(q,s)$ entanglement of assistance for $0 \leq \beta \leq 1$.
Finally, we summarize our results in Section~\ref{Sec: Conclusion}.

\section{Unified entropy and multi-qubit entanglement constraints}
\label{Sec: unimeasure}

For $q,~s \geq 0$ with $q \neq 1$ and $s \neq 0$, unified-$(q,s)$ entropy of a quantum state $\rho$ is defined as~\cite{ue1, ue2},
\begin{equation}
S_{q,s}(\rho):=\frac{1}{(1-q)s}\left[{\left(\T \rho^{q}\right)}^s-1\right].
\label{uqs-entropy}
\end{equation}
Although unified-$(q,s)$ entropy has a singularity at $s=0$, it
converges to R\'enyi-$q$ entropy as $s$ tends to 0~\cite{renyi, horo}.
We also note that unified-$(q,s)$ entropy converges to Tsallis-$q$ entropy ~\cite{tsallis} when $s$ tends to $1$,
and for any nonnegative $s$, unified-$(q,s)$ entropy converges
to von Neumann entropy as $q$ tends to 1,
\begin{align}
\lim_{q \rightarrow 1}S_{q,s}(\rho)=-\T \rho\log\rho 
=:S(\rho), \label{T1}
\end{align}
and these enable us to denote $S_{1,s}(\rho)=S(\rho)$ and $S_{q,0}(\rho)=R_{q}(\rho)$.

Using unified-$(q,s)$ entropy in Eq.~(\ref{uqs-entropy}), a two-parameter class of bipartite entanglement measures was introduced; for a bipartite pure state $\ket{\psi}_{AB}$, its {\em unified-$(q,s)$ entanglement}(UE)~\cite{KSU} is
\begin{equation}
E_{q,s}\left(\ket{\psi}_{A|B} \right):=S_{q,s}(\rho_A),
\label{UEpure}
\end{equation}
for each $q,~s \geq 0$ where $\rho_A=\T _{B} \ket{\psi}_{AB}\bra{\psi}$
is the reduced density matrix of $\ket{\psi}_{AB}$ onto subsystem $A$.
For a bipartite mixed state $\rho_{AB}$, its UE is
\begin{equation}
E_{q,s}\left(\rho_{A|B} \right):=\min \sum_i p_i
E_{q,s}(\ket{\psi_i}_{A|B}), \label{UEmixed}
\end{equation}
where the minimum is taken over all possible pure state
decompositions of $\rho_{AB}=\sum_{i}p_i
\ket{\psi_i}_{AB}\bra{\psi_i}$.
As a dual concept to UE, {\em unified-$(q,s)$ entanglement of assistance}(UEoA) was also introduced as
\begin{equation}
E^a_{q,s}\left(\rho_{A|B} \right):=\max \sum_i p_i
E_{q,s}(\ket{\psi_i}_{A|B}), \label{UEoA}
\end{equation}
for $q,~s \geq 0$ where the maximum is taken over all possible pure state decompositions of $\rho_{AB}$~\cite{KUP}.

Due to the continuity of UE in Eq.~(\ref{UEmixed}) with respect to the parameters $q$ and $s$, UE reduces to the one-parameter class of entanglement measures namely {\em R\'enyi-$q$ entanglement}(RE)~\cite{KSRenyi} as $s$ tends to $0$,
and it also reduces to another one-parameter
class of bipartite entanglement measures called {\em Tsallis-$q$ entanglement}(TE)~\cite{KimT} as $s$ tends to $1$.
For any nonnegative $s$, UE converges to {\em entanglement of formation}(EoF) as $q$ tends to 1,
\begin{align}
\lim_{q\rightarrow1}E_{q,s}\left(\rho_{A|B} \right)=E_{\rm f}\left(\rho_{A|B} \right),
\end{align}
therefore UE is one of the most general classes of bipartite entanglement measures including the classes of R\'enyi and Tsallis entanglements and EoF as special cases~\cite{KSU}.

Similarly, the continuity of UEoA in Eq.~(\ref{UEoA}) with respect to the parameters $q$ and $s$ assures that UEoA reduces to {\em R\'enyi-$q$ entanglement of assistance}(REoA)~\cite{KSRenyi} and {\em Tsallis-$q$ entanglement of assistance}(TEoA)~\cite{KimT} when $s$ tends to $0$ or $1$ respectively. For any nonnegative $s$, UEoA reduces to {\em entanglement of assistance}(EoA)
\begin{align}
\lim_{q\rightarrow1}{E^a_{q,s}}\left(\rho_{AB}
\right)=E^a\left(\rho_{AB} \right), \label{UE_EoA}
\end{align}
when $q$ tends to $1$~\cite{KUP}.

Using UE as the bipartite entanglement measure, a two-parameter class of monogamy inequalities of multi-qubit entanglement was established~\cite{KSU}; for $q\geq2$, $0\leq s \leq1$ and $qs\leq3$, we have
\begin{align}
E_{q,s}\left( \rho_{A_1|A_2 \cdots A_N}\right)\geq \sum_{i=2}^{N}E_{q,s}\left(\rho_{A_1|A_i}\right)
\label{Umono}
\end{align}
for any multi-qubit state $\rho_{A_1 \cdots A_N}$ where $E_{q,s}\left( \rho_{A_1|A_2 \cdots A_N}\right)$ is the
UE of $\rho_{A_1A_2 \cdots
A_N}$ with respect to the bipartition between $A_1$ and
$A_{2}\cdots A_{N}$, and $E_{q,s}\left(\rho_{A_1| A_i}\right)$ is the
unified-$(q,s)$ entanglement of the reduced density matrix
$\rho_{A_1 A_i}$ for each $i=2,\cdots,N$. \label{Prop: UEmono}

Later, it was shown that unified entropy can also be used to establish a class of polygamy inequalities of multi-qubit entanglement~\cite{KUP};
for $1 \leq q \leq 2$ and $-q^2+4q-3 \leq s \leq1$, we have
\begin{equation}
E^a_{q,s}\left( \rho_{A_1|A_2 \cdots A_N}\right)\leq
\sum_{i=2}^{N}E^a_{q,s}\left(\rho_{A_1|A_i}\right)
\label{UEpoly}
\end{equation}
for any multi-qubit state $\rho_{A_1 \cdots A_N}$ where $E^a_{q,s}\left( \rho_{A_1|A_2 \cdots A_N}\right)$ is the
UEoA of $\rho_{A_1A_2 \cdots A_N}$ with
respect to the bipartition between $A_1$ and $A_{2}\cdots A_{N}$,
and $E^a_{q,s}(\rho_{A_1|A_i})$ is the UEoA of $\rho_{A_1 A_i}$ for $i=2,\cdots,N$. \label{Thm: poly}

\section{Tight monogamy constraints of multi-qubit entanglement in terms of unified entanglement}
\label{Sec: TM_qubit}

Here we establish a class of tight monogamy inequalities of multi-qubit entanglement using the $\alpha$'th power of UE. Before we present our main results, we first provide some notations, definitions and a lemma, which are useful throughout this paper.

For any nonnegative integer $j$ whose binary expansion is
\begin{align}
j=\sum_{i=0}^{n-1} j_i 2^i
\label{bire2}
\end{align}
where $\log_{2}j \leq n$ and $j_i \in \{0, 1\}$ for $i=0, \cdots, n-1$,
we can always define a unique binary vector
associated with $j$, which is defined as
\begin{align}
\overrightarrow{j}=\left(j_0,~j_1,~\cdots ,j_{n-1}\right).
\label{bivec}
\end{align}
For the binary vector $\overrightarrow{j}$ in  Eq.~(\ref{bivec}), its {\em Hamming weight}, $\omega_{H}\left(\overrightarrow{j}\right)$, is the number of $1's$ in its coordinates~\cite{nc}.
We also provide the following lemma whose proof is easily obtained by some straightforward calculus.
\begin{Lem}
For $x \in \left[0,1\right]$ and nonnegative real numbers $\alpha, \beta$, we have
\begin{align}
\left(1+x\right)^{\alpha}\geq1+\alpha x^{\alpha}
\label{alpgre1}
\end{align}
for $\alpha \geq1$, and
\begin{align}
\left(1+x\right)^{\beta}\leq 1+\beta x^{\beta}
\label{betle1}
\end{align}
for $0\leq\beta \leq1$.
\label{Lem: ineqs}
\end{Lem}

Now we provide our first result, which states that a class of tight monogamy inequalities of multi-qubit entanglement can be established using the $\alpha$-powered UE and the Hamming weight of the binary vector related with the distribution of subsystems.
\begin{Thm}
For $\alpha \geq 1$, $q\geq2$ and $0\leq s \leq1$, $qs\leq3$, we have
\begin{equation}
\left(E_{q,s}\left(\rho_{A|B_0 B_1 \cdots B_{N-1}}\right)\right)^{\alpha}  \geq
\sum_{j=0}^{N-1}{\alpha}^{\omega_{H}\left(\overrightarrow{j}\right)}\left(E_{q,s}\left(\rho_{A|B_j}\right)\right)^{\alpha},
\label{pcrenmono0}
\end{equation}
for any multi-qubit state $\rho_{AB_0\cdots B_{N-1}}$ where $\overrightarrow{j}=\left(j_0, \cdots ,j_{n-1}\right)$ is the vector from the binary representation of $j$ and
$\omega_{H}\left(\overrightarrow{j}\right)$ is the Hamming weight of $\overrightarrow{j}$.
\label{tmono0}
\end{Thm}

\begin{proof}
Without loss of generality, we may assume that the ordering of the qubit subsystems $B_0, \ldots, B_{N-1}$ satisfies
\begin{align}
E_{q,s}\left(\rho_{A |B_j}\right)\geq E_{q,s}\left(\rho_{A |B_{j+1}}\right)\geq0
\label{pordern}
\end{align}
for each $j=0, \cdots , N-2$ by reordering and relabeling them, if necessary.

From the monotonicity of the function $f(x)=x^{\alpha}$ for $\alpha \geq 1$ and the UE-based monogamy inequality of multi-qubit entanglement in~(\ref{Umono}), we have
\begin{align}
\left(E_{q,s}\left(\rho_{A|B_0B_1\cdots B_{N-1}}\right)\right)^{\alpha}\geq&
\left(\sum_{j=0}^{N-1}E_{q,s}\left(\rho_{A|B_j}\right)\right)^{\alpha},
\label{qbitalpmono1}
\end{align}
which makes it feasible to prove the theorem by showing
\begin{align}
\left(\sum_{j=0}^{N-1}E_{q,s}\left(\rho_{A|B_j}\right)\right)^{\alpha}\geq&
\sum_{j=0}^{N-1}{\alpha}^{\omega_{H}\left(\overrightarrow{j}\right)}\left(E_{q,s}\left(\rho_{A|B_j}\right)\right)^{\alpha}.
\label{qbitalpmono2}
\end{align}
We first prove Inequality~(\ref{qbitalpmono2}) for the case that $N=2^n$, a power of 2, by using mathematical induction on $n$, and extend the result for any positive integer $N$.

For $n=1$ and a three-qubit state $\rho_{AB_0B_1}$ with two-qubit rduced density matrices $\rho_{AB_0}$ and $\rho_{AB_1}$, we have
\begin{align}
&\left(E_{q,s}\left(\rho_{A|B_0}\right)+E_{q,s}\left(\rho_{A|B_1}\right)\right)^{\alpha}
=\left(E_{q,s}\left(\rho_{A|B_0}\right)\right)^{\alpha}\left(1+\frac{E_{q,s}\left(\rho_{A|B_1}\right)}
{E_{q,s}\left(\rho_{A|B_0}\right)} \right)^{\alpha},
\label{p3scmono1}
\end{align}
where Inequalities~(\ref{alpgre1}) and (\ref{pordern}) implies
\begin{align}
\left(1+\frac{E_{q,s}\left(\rho_{A|B_1}\right)}{E_{q,s}\left(\rho_{A|B_0}\right)} \right)^{\alpha}
\geq1+\alpha\left(\frac{E_{q,s}\left(\rho_{A|B_1}\right)}
{E_{q,s}\left(\rho_{A|B_0}\right)}\right)^{\alpha}.
\label{p3scmono2}
\end{align}
From Eq.~(\ref{p3scmono1}) and Inequality~(\ref{p3scmono2}), we have
\begin{align}
\left(E_{q,s}\left(\rho_{A|B_0}\right)+E_{q,s}\left(\rho_{A|B_1}\right)\right)^{\alpha}\geq
\left(E_{q,s}\left(\rho_{A|B_0}\right)\right)^{\alpha}+
\alpha\left(E_{q,s}\left(\rho_{A|B_1}\right)\right)^{\alpha},
\label{p3scmono3}
\end{align}
which recovers Inequality~(\ref{qbitalpmono2}) for $n=1$.

Now let us assume Inequality~(\ref{qbitalpmono2}) is true for $N=2^{n-1}$ with $n\geq 2$, and consider the case that $N=2^n$.
For an $(N+1)$-qubit state $\rho_{AB_0B_1 \cdots B_{N-1}}$ with its two-qubit reduced density matrices $\rho_{AB_j}$ with $j=0, \cdots, N-1$,
we have
\begin{align}
\left(\sum_{j=0}^{N-1}E_{q,s}\left(\rho_{A|B_j}\right)\right)^{\alpha}=
\left(\sum_{j=0}^{2^{n-1}-1}E_{q,s}\left(\rho_{A|B_j}\right)\right)^{\alpha}
\left(1+\frac{\sum_{j=2^{n-1}}^{2^n-1}E_{q,s}\left(\rho_{A|B_j}\right)}
{\sum_{j=0}^{2^{n-1}-1}E_{q,s}\left(\rho_{A|B_j}\right)}\right)^{\alpha}.
\label{pnscmono}
\end{align}
Because the ordering of subsystems in Inequality~(\ref{pordern}) implies
\begin{align}
0\leq\frac{\sum_{j=2^{n-1}}^{2^n-1}E_{q,s}\left(\rho_{A|B_j}\right)}
{\sum_{j=0}^{2^{n-1}-1}E_{q,s}\left(\rho_{A|B_j}\right)}\leq 1,
\label{sumine1}
\end{align}
thus Eq.~(\ref{pnscmono}) and Inequality~(\ref{alpgre1}) lead us to
\begin{align}
\left(\sum_{j=0}^{N-1}E_{q,s}\left(\rho_{A|B_j}\right)\right)^{\alpha}
\geq
\left(\sum_{j=0}^{2^{n-1}-1}E_{q,s}\left(\rho_{A|B_j}\right)\right)^{\alpha}+
\alpha\left(\sum_{j=2^{n-1}}^{2^{n}-1}E_{q,s}\left(\rho_{A|B_j}\right)\right)^{\alpha}.
\label{pnscmono1}
\end{align}

From the induction hypothesis, we have
\begin{align}
\left(\sum_{j=0}^{2^{n-1}-1}E_{q,s}\left(\rho_{A|B_j}\right)\right)^{\alpha}\geq&
\sum_{j=0}^{2^{n-1}-1}{\alpha}^{\omega_{H}\left(\overrightarrow{j}\right)}\left(E_{q,s}\left(\rho_{A|B_j}\right)\right)^{\alpha}.
\label{pnscmono2}
\end{align}
Moreover, the last summation in Inequality~(\ref{pnscmono1}) is also a summation of $2^{n-1}$ terms starting from $j=2^{n-1}$ to $j=2^{n}-1$.
Thus, (after possible indexing and reindexing subsystems, if necessary) the induction hypothesis also leads us to
\begin{align}
\left(\sum_{j=2^{n-1}}^{2^{n}-1}E_{q,s}\left(\rho_{A|B_j}\right)\right)^{\alpha}\geq&
\sum_{j=2^{n-1}}^{2^{n}-1}{\alpha}^{\omega_{H}\left(\overrightarrow{j}\right)-1}\left(E_{q,s}\left(\rho_{A|B_j}\right)\right)^{\alpha}.
\label{pnscmono3}
\end{align}
From Inequalities~(\ref{pnscmono1}), (\ref{pnscmono2}) and (\ref{pnscmono3}), we have
\begin{align}
\left(\sum_{j=0}^{2^n-1}E_{q,s}\left(\rho_{A|B_j}\right)\right)^{\alpha}
\geq& \sum_{j=0}^{2^n-1}{\alpha}^{\omega_{H}\left(\overrightarrow{j}\right)}\left(E_{q,s}\left(\rho_{A|B_j}\right)\right)^{\alpha},
\label{pnscmono4}
\end{align}
which recovers Inequality~(\ref{qbitalpmono2}) for $N=2^n$.

Now let us consider an arbitrary positive integer $N$ and a $(N+1)$-qubit state $\rho_{AB_0B_1\cdots B_{N-1}}$. We first note that
we can always consider a power of $2$, which is an upper bound of $N$, that is, $0\leq N \leq 2^{n}$ for some $n$.
We also consider a $(2^{n}+1)$-qubit state
\begin{align}
\Gamma_{AB_0 B_1 \cdots B_{2^n-1}}=\rho_{AB_0B_1\cdots B_{N-1}}\otimes\sigma_{B_N \cdots B_{2^n-1}},
\label{gamma}
\end{align}
which is a product of $\rho_{AB_0B_1\cdots B_{N-1}}$ and an arbitrary $(2^n-N)$-qubit state $\sigma_{B_N \cdots B_{2^n-1}}$.

Because $\Gamma_{AB_0 B_1 \cdots B_{2^n-1}}$ is a $(2^{n}+1)$-qubit state, Inequality~(\ref{pnscmono4}) leads us to
\begin{equation}
\left(E_{q,s}\left(\Gamma_{A|B_0 B_1 \cdots B_{2^n-1}}\right)\right)^{\alpha}  \geq
\sum_{j=0}^{2^n-1}{\alpha}^{\omega_{H}\left(\overrightarrow{j}\right)}\left(E_{q,s}\left(\Gamma_{A|B_j}\right)\right)^{\alpha},
\label{gamono}
\end{equation}
where $\Gamma_{AB_j}$ is the two-qubit reduced density matric of $\Gamma_{AB_0 B_1 \cdots B_{2^n-1}}$ for each $j= 0, \cdots, 2^n-1$.
On the other hand, the separability of $\Gamma_{AB_0 B_1 \cdots B_{2^n-1}}$ with respect to the bipartition between $AB_0\cdots B_{N-1}$ and $B_N \cdots B_{2^n-1}$ assures
\begin{align}
E_{q,s}\left(\Gamma_{A|B_0 B_1 \cdots B_{2^n-1}}\right)=E_{q,s}\left(\rho_{A|B_0 B_1 \cdots B_{N-1}}\right),
\label{same1}
\end{align}
as well as
\begin{align}
E_{q,s}\left(\Gamma_{A|B_j}\right)=0,
\label{same3}
\end{align}
for $j=N, \cdots , 2^n-1$.
Moreover, we have
\begin{align}
\Gamma_{AB_j}=\rho_{AB_j},
\label{same2}
\end{align}
for each $j=0, \cdots , N-1$.
Thus, Inequality~(\ref{gamono}) together with Eqs.~(\ref{same1}), (\ref{same3}) and (\ref{same2}) leads us to
\begin{align}
\left(E_{q,s}\left(\rho_{A|B_0 B_1 \cdots B_{N-1}}\right)  \right)^{\alpha}=&
\left(E_{q,s}\left(\Gamma_{A|B_0 B_1 \cdots B_{2^n-1}}\right)\right)^{\alpha}\nonumber\\
\geq&\sum_{j=0}^{2^n-1}{\alpha}^{\omega_{H}\left(\overrightarrow{j}\right)}\left(E_{q,s}\left(\Gamma_{A|B_j}\right)\right)^{\alpha}\nonumber\\
=&\sum_{j=0}^{N-1}{\alpha}^{\omega_{H}\left(\overrightarrow{j}\right)}\left(E_{q,s}\left(\rho_{A|B_j}\right)\right)^{\alpha},
\label{pnscmono5}
\end{align}
and this completes the proof.
\end{proof}

For any $\alpha\geq 1$ and the Hamming weight $\omega_{H}\left(\overrightarrow{j}\right)$ of the binary vector $\overrightarrow{j}=\left(j_0, \cdots ,j_{n-1}\right)$, ${\alpha}^{\omega_{H}\left(\overrightarrow{j}\right)}$ is greater than or equal to $1$, therefore
\begin{align}
\left(E_{q,s}\left(\rho_{A|B_0 B_1 \cdots B_{N-1}}\right)\right)^{\alpha}  \geq
\sum_{j=0}^{N-1}{\alpha}^{\omega_{H}\left(\overrightarrow{j}\right)}\left(E_{q,s}\left(\rho_{A|B_j}\right)\right)^{\alpha}
\geq\sum_{j=0}^{N-1}\left(E_{q,s}\left(\rho_{A|B_j}\right)\right)^{\alpha},
\label{ineqtight1}
\end{align}
for any multi-qubit state $\rho_{AB_0 B_1 \cdots B_{N-1}}$ and $\alpha\geq 1$. Thus Inequality~(\ref{pcrenmono0}) of
Theorem~\ref{tmono0} is generally tighter than the monogamy inequalities of multi-qubit entanglement, which just use
the $\alpha$th-power of entanglement measures.

Due to the continuity of UE with respect to the parameters $q$ and $s$, Inequality~(\ref{pcrenmono0}) of
Theorem~\ref{tmono0} reduces to the class of R\'enyi-$q$ entropy-based monogamy inequalities of multi-qubit entanglement~\cite{KSRenyi} in a tighter way when $s$ tends to $0$;
\begin{equation}
\left({\mathcal R}_{q}\left(\rho_{A|B_0 B_1 \cdots B_{N-1}}\right)\right)^{\alpha}  \geq
\sum_{j=0}^{N-1}{\alpha}^{\omega_{H}\left(\overrightarrow{j}\right)}\left({\mathcal R}_{q}\left(\rho_{A|B_j}\right)\right)^{\alpha},
\label{renyimono0}
\end{equation}
for any $\alpha \geq1$, $q\geq2$ and multi-qubit state $\rho_{AB_0 B_1 \cdots B_{N-1}}$ where ${\mathcal R}_{q}\left(\rho_{A|B_0 B_1 \cdots B_{N-1}}\right)$ is the RE of $\rho_{AB_0 B_1 \cdots B_{N-1}}$ with respect to the bipartition between $A$ and $B_0 B_1 \cdots B_{N-1}$~\cite{KSRenyi}.
When $s$ tends to $1$, Inequality~(\ref{pcrenmono0}) reduces to another class of
monogamy inequalities, namely, Tsallis-$q$ entropy-based monogamy inequalities of multi-qubit entanglement~\cite{KimT} in a tighter way;
\begin{equation}
\left({\mathcal T}_{q}\left(\rho_{A|B_0 B_1 \cdots B_{N-1}}\right)\right)^{\alpha}  \geq
\sum_{j=0}^{N-1}{\alpha}^{\omega_{H}\left(\overrightarrow{j}\right)}\left({\mathcal T}_{q}\left(\rho_{A|B_j}\right)\right)^{\alpha},
\label{Tmono0}
\end{equation}
for any $\alpha \geq1$, $2 \leq q \leq 3$ and multi-qubit state $\rho_{AB_0 B_1 \cdots B_{N-1}}$ where ${\mathcal T}_{q}\left(\rho_{A|B_0 B_1 \cdots B_{N-1}}\right)$ is the TE of $\rho_{AB_0 B_1 \cdots B_{N-1}}$ with respect to the bipartition between $A$ and $B_0 B_1 \cdots B_{N-1}$~\cite{KimT}.

We also note that Inequality~(\ref{pcrenmono0}) of Theorem~\ref{tmono0} can be even improved to be a tighter inequality with some condition on two-qubit entanglement;
\begin{Thm}
For $\alpha \geq 1$, $q\geq2$, $0\leq s \leq1$, $qs\leq3$ and any multi-qubit state $\rho_{AB_0\cdots B_{N-1}}$, we have
\begin{equation}
\left(E_{q,s}\left(\rho_{A|B_0 \cdots B_{N-1}}\right)\right)^{\alpha}  \geq
\sum_{j=0}^{N-1}{\alpha}^{j}\left(E_{q,s}\left(\rho_{A|B_j}\right)\right)^{\alpha},
\label{pcrenmono2}
\end{equation}
conditioned that
\begin{align}
E_{q,s}\left(\rho_{A|B_i}\right)\geq \sum_{j=i+1}^{N-1}E_{q,s}\left(\rho_{A|B_{j}}\right),
\label{cond2}
\end{align}
for $i=0, \cdots , N-2$.
\label{tmono2}
\end{Thm}

\begin{proof}
Due to Inequality~(\ref{qbitalpmono1}), it is enough to show
\begin{align}
\left(\sum_{j=0}^{N-1}E_{q,s}\left(\rho_{A|B_j}\right)\right)^{\alpha}\geq&
\sum_{j=0}^{N-1}{\alpha}^{j}\left(E_{q,s}\left(\rho_{A|B_j}\right)\right)^{\alpha},
\label{qbitalpmono3}
\end{align}
and we use mathematical induction on $N$.
We further note that Inequality~(\ref{p3scmono3})
in the proof of Theorem~\ref{tmono0} assures that Inequality~(\ref{qbitalpmono3}) is true for $N=2$.

Now let us assume the validity of Inequality~(\ref{qbitalpmono3})
for any positive integer less than $N$.
For a multi-qubit state $\rho_{AB_0 \cdots B_{N-1}}$, we have
\begin{align}
\left(\sum_{j=0}^{N-1}E_{q,s}\left(\rho_{A|B_j}\right)\right)^{\alpha}
=&\left(E_{q,s}\left(\rho_{A|B_0}\right)\right)^{\alpha}
\left(1+\frac{\sum_{j=1}^{N-1}E_{q,s}\left(\rho_{A|B_j}\right)}
{E_{q,s}\left(\rho_{A|B_0}\right)} \right)^{\alpha},
\label{pnscmono5}
\end{align}
where Inequality~(\ref{alpgre1}) and the condition in Inequality~(\ref{cond2}) lead Inequality~(\ref{pnscmono5}) to
\begin{align}
\left(1+\frac{\sum_{j=1}^{N-1}E_{q,s}\left(\rho_{A|B_j}\right)}
{E_{q,s}\left(\rho_{A|B_0}\right)} \right)^{\alpha}\geq&1
+\alpha\left(\frac{\sum_{j=1}^{N-1}E_{q,s}\left(\rho_{A|B_j}\right)}
{E_{q,s}\left(\rho_{A|B_0}\right)}\right)^{\alpha}.
\label{coninq1}
\end{align}
Thus Eq.~(\ref{pnscmono5}) and Inequality~(\ref{coninq1}) imply
\begin{align}
\left(\sum_{j=0}^{N-1}E_{q,s}\left(\rho_{A|B_j}\right)
\right)^{\alpha}\geq&
\left(E_{q,s}\left(\rho_{A|B_0}\right)\right)^{\alpha}
+\alpha\left(\sum_{j=1}^{N-1}E_{q,s}\left(\rho_{A|B_j}\right)
\right)^{\alpha}.
\label{pnscmono6}
\end{align}

Because the summation in the right-hand side of Inequality~(\ref{pnscmono6}) is a summation of $N-1$ terms, the induction hypothesis assures
\begin{align}
\left(\sum_{j=1}^{N-1}E_{q,s}\left(\rho_{A|B_j}\right)
\right)^{\alpha}\geq
\sum_{j=1}^{N-1}{\alpha}^{j-1}\left(E_{q,s}\left(\rho_{A|B_j}\right)\right)^{\alpha}.
\label{pnscmono7}
\end{align}
Now, Inequality~(\ref{pnscmono6}) together with Inequality~(\ref{pnscmono7}) recover Inequality~(\ref{qbitalpmono3}),
and this completes the proof.
\end{proof}

For any nonnegative integer $j$ and its corresponding binary vector $\overrightarrow{j}$, the Hamming weight $\omega_{H}\left(\overrightarrow{j}\right)$ is bounded above by $\log_2 j$. Thus we have
\begin{align}
\omega_{H}\left(\overrightarrow{j}\right)\leq \log_2 j \leq j,
\label{numcom}
\end{align}
which implies
\begin{align}
\left(E_{q,s}\left(\rho_{A|B_0 \cdots B_{N-1}}\right)\right)^{\alpha}  \geq&
\sum_{j=0}^{N-1}{\alpha}^{j}\left(E_{q,s}\left(\rho_{A|B_j}\right)\right)^{\alpha}
\geq \sum_{j=0}^{N-1}{\alpha}^{\omega_{H}\left(\overrightarrow{j}\right)}\left(E_{q,s}\left(\rho_{A|B_j}\right)\right)^{\alpha},
\label{ineqtight2}
\end{align}
for any $\alpha\geq 1$. In other words, Inequality~(\ref{pcrenmono0}) in
Theorem~\ref{tmono0} can be made to be even tighter as
Inequality~(\ref{pcrenmono2}) of Theorem~\ref{tmono2} for any multi-qubit state $\rho_{AB_0 B_1 \cdots B_{N-1}}$ satisfying the condition in Inequality~(\ref{cond2}).

\section{Tight polygamy constraints of multi-qubit entanglement in terms of unified entanglement of assistance}
\label{Sec: TP_qubit}

As a dual property to the Inequality~(\ref{pcrenmono0}) of
Theorem~\ref{tmono0}, we provide a class of polygamy inequalities of multi-qubit entanglement in terms of
powered UEoA.
\begin{Thm}
For~$0\leq \beta\leq 1$, $-q^2+4q-3 \leq s \leq 1$ on $1 \leq q \leq 2$ and any multi-qubit state $\rho_{AB_0\cdots B_{N-1}}$, we have
\begin{equation}
\left(E^a_{q,s}\left(\rho_{A|B_0 B_1 \cdots B_{N-1}}\right)\right)^{\beta}  \leq
\sum_{j=0}^{N-1}{\beta}^{\omega_{H}\left(\overrightarrow{j}\right)}\left(E^a_{q,s}\left(\rho_{A|B_j}\right)\right)^{\beta}.
\label{pcrenpoly0}
\end{equation}
\label{tpoly0}
\end{Thm}

\begin{proof}
Without loss of generality, we assume the ordering of the qubit subsystems $B_0, \cdots, B_{N-1}$ satisfying
\begin{align}
E^a_{q,s}\left(\rho_{A |B_j}\right)\geq E^a_{q,s}\left(\rho_{A |B_{j+1}}\right)\geq0
\label{pordern2}
\end{align}
for each $j=0, \cdots , N-2$.
Moreover, due to the monotonicity of the function $f(x)=x^{\beta}$ for $0 \leq \beta \leq 1$ and the UEoA-based multi-qubit polygamy inequality in~(\ref{UEpoly}), we have
\begin{align}
\left(E^a_{q,s}\left(\rho_{A|B_0B_1\cdots B_{N-1}}\right)\right)^{\beta}\leq&
\left(\sum_{j=0}^{N-1}E^a_{q,s}\left(\rho_{A|B_j}\right)\right)^{\beta},
\label{qbitalppoly1}
\end{align}
thus it is enough to show that
\begin{align}
\left(\sum_{j=0}^{N-1}E^a_{q,s}\left(\rho_{A|B_j}\right)\right)^{\beta}\leq&
\sum_{j=0}^{N-1}{\beta}^{\omega_{H}\left(\overrightarrow{j}\right)}\left(E_{q,s}\left(\rho_{A|B_j}\right)\right)^{\beta}.
\label{qbitalppoly2}
\end{align}

The proof method is similar to that of Theorem~\ref{tmono0};
we first prove Inequality~(\ref{qbitalppoly2}) for the case that $N=2^n$ by using mathematical induction on $n$, and generalize the result to any positive integer $N$.
For $n=1$ and a three-qubit state $\rho_{AB_0B_1}$ with two-qubit rduced density matrices $\rho_{AB_0}$ and $\rho_{AB_1}$, we have
\begin{align}
&\left(E^a_{q,s}\left(\rho_{A|B_0}\right)+E^a_{q,s}\left(\rho_{A|B_1}\right)\right)^{\beta}
=\left(E^a_{q,s}\left(\rho_{A|B_0}\right)\right)^{\beta}\left(1+\frac{E^a_{q,s}\left(\rho_{A|B_1}\right)}
{E^a_{q,s}\left(\rho_{A|B_0}\right)} \right)^{\beta},
\label{p3scpoly1}
\end{align}
which, together with Inequalities~(\ref{betle1}) and (\ref{pordern2}) leads us to
\begin{align}
&\left(E^a_{q,s}\left(\rho_{A|B_0}\right)+E^a_{q,s}\left(\rho_{A|B_1}\right)\right)^{\beta}
\leq\left(E^a_{q,s}\left(\rho_{A|B_0}\right)\right)^{\beta}+
\beta\left(E^a_{q,s}\left(\rho_{A|B_1}\right)\right)^{\beta}.
\label{p3scpoly3}
\end{align}
Inequality~(\ref{p3scpoly3}) recovers Inequality~(\ref{qbitalppoly2}) for $n=1$.

Now we assume the validity of Inequality~(\ref{qbitalppoly2}) for $N=2^{n-1}$ with $n\geq 2$, and consider the case that $N=2^n$.
For an $(N+1)$-qubit state $\rho_{AB_0B_1 \cdots B_{N-1}}$ and its two-qubit reduced density matrices $\rho_{AB_j}$ with $j=0, \cdots, N-1$,
we have
\begin{align}
\left(\sum_{j=0}^{N-1}E^a_{q,s}\left(\rho_{A|B_j}\right)\right)^{\beta}
=&\left(\sum_{j=0}^{2^{n-1}-1}E^a_{q,s}\left(\rho_{A|B_j}\right)\right)^{\beta}
\left(1+\frac{\sum_{j=2^{n-1}}^{2^n-1}E^a_{q,s}\left(\rho_{A|B_j}\right)}
{\sum_{j=0}^{2^{n-1}-1}E^a_{q,s}\left(\rho_{A|B_j}\right)}\right)^{\beta},
\label{pnscpoly}
\end{align}
where the ordering of subsystems in Inequality~(\ref{pordern2}) and
Inequality~(\ref{betle1}) together with Eq.~(\ref{pnscpoly}) lead us to
\begin{align}
\left(\sum_{j=0}^{N-1}E^a_{q,s}\left(\rho_{A|B_j}\right)\right)^{\beta}
\leq&\left(\sum_{j=0}^{2^{n-1}-1}E^a_{q,s}\left(\rho_{A|B_j}\right)\right)^{\beta}
+\beta\left(\sum_{j=2^{n-1}}^{2^{n}-1}E^a_{q,s}\left(\rho_{A|B_j}\right)\right)^{\beta}.
\label{pnscpoly1}
\end{align}

Because each summation on the right-hand side of Inequality~(\ref{pnscpoly1}) is a summation of $2^{n-1}$ terms , the induction hypothesis assures that
\begin{align}
\left(\sum_{j=0}^{2^{n-1}-1}E^a_{q,s}\left(\rho_{A|B_j}\right)\right)^{\beta}\leq&
\sum_{j=0}^{2^{n-1}-1}{\beta}^{\omega_{H}\left(\overrightarrow{j}\right)}\left(E^a_{q,s}\left(\rho_{A|B_j}\right)\right)^{\beta},
\label{pnscpoly2}
\end{align}
and
\begin{align}
\left(\sum_{j=2^{n-1}}^{2^{n}-1}E^a_{q,s}\left(\rho_{A|B_j}\right)\right)^{\beta}\leq&
\sum_{j=2^{n-1}}^{2^{n}-1}{\beta}^{\omega_{H}\left(\overrightarrow{j}\right)-1}\left(E^a_{q,s}\left(\rho_{A|B_j}\right)\right)^{\beta}.
\label{pnscpoly3}
\end{align}
(Possibly, we may index and reindex subsystems to get Inequality~(\ref{pnscpoly3}), if necessary.)
Thus, Inequalities~(\ref{pnscpoly1}),~ (\ref{pnscpoly2}) and (\ref{pnscpoly3})
recover Inequality~(\ref{qbitalppoly2}) when $N=2^n$.

For an arbitrary positive integer $N$ and a $(N+1)$-qubit state $\rho_{AB_0B_1\cdots B_{N-1}}$,
leu us consider the $(2^{n}+1)$-qubit state $\Gamma_{AB_0 B_1 \cdots B_{2^n-1}}$
in Eq.~(\ref{gamma}). Because $\Gamma_{AB_0 B_1 \cdots B_{2^n-1}}$ is a $(2^{n}+1)$-qubit state, we have
\begin{equation}
\left(E^a_{q,s}\left(\Gamma_{A|B_0 B_1 \cdots B_{2^n-1}}\right)\right)^{\beta}\leq
\sum_{j=0}^{2^n-1}{\beta}^{\omega_{H}\left(\overrightarrow{j}\right)}\left(E^a_{q,s}\left(\Gamma_{A|B_j}\right)\right)^{\beta},
\label{gapoly}
\end{equation}
where $\Gamma_{AB_j}$ is the two-qubit reduced density matric of $\Gamma_{AB_0 B_1 \cdots B_{2^n-1}}$ for each $j= 0, \cdots, 2^n-1$.

Moreover, $\Gamma_{AB_0 B_1 \cdots B_{2^n-1}}$ is a product state of $\rho_{AB_0B_1\cdots B_{N-1}}$ and $\sigma_{B_N \cdots B_{2^n-1}}$,
which implies
\begin{align}
E^a_{q,s}\left(\Gamma_{A|B_0 B_1 \cdots B_{2^n-1}}\right)=E^a_{q,s}\left(\rho_{A|B_0 B_1 \cdots B_{N-1}}\right),
\label{psame1}
\end{align}
and
\begin{align}
E^a_{q,s}\left(\Gamma_{A|B_j}\right)=0,
\label{psame3}
\end{align}
for $j=N, \cdots , 2^n-1$. We also note that
\begin{align}
\Gamma_{AB_j}=\rho_{AB_j},
\label{psame2}
\end{align}
for each $j=0, \cdots , N-1$. Thus Inequality~(\ref{gapoly}) together with Eqs.~(\ref{psame1}), (\ref{psame3}) and (\ref{psame2}) recovers
Inequality~(\ref{pcrenpoly0}), and this completes the proof.
\end{proof}

Similarly to the case of monogamy inequalities, Inequality~(\ref{pcrenpoly0}) of Theorem~\ref{tpoly0}
reduces to a class of Tsallis-$q$ entropy-based polygamy inequalities of multi-qubit entanglement in a tighter way;
\begin{equation}
\left({\mathcal T}^a_{q}\left(\rho_{A|B_0 B_1 \cdots B_{N-1}}\right)\right)^{\beta}  \leq
\sum_{j=0}^{N-1}{\beta}^{\omega_{H}\left(\overrightarrow{j}\right)}\left({\mathcal T}^a_{q}\left(\rho_{A|B_j}\right)\right)^{\beta},
\label{Tpoly0}
\end{equation}
for any $0\leq \beta\leq 1$, $1 \leq q \leq 2$ and multi-qubit state $\rho_{AB_0 B_1 \cdots B_{N-1}}$ where ${\mathcal T}^a_{q}\left(\rho_{A|B_0 B_1 \cdots B_{N-1}}\right)$ is the TEoA of $\rho_{AB_0 B_1 \cdots B_{N-1}}$ with respect to the bipartition between $A$ and $B_0 B_1 \cdots B_{N-1}$~\cite{KimT}. We further note that Inequality~(\ref{pcrenpoly0}) of Theorem~\ref{tpoly0} can be improved to a class of tighter
polygamy inequalities with some condition on two-qubit entanglement of assistance.

\begin{Thm}
For $0\leq \beta\leq 1$, $-q^2+4q-3 \leq s \leq 1$ on $1 \leq q \leq 2$ and any multi-qubit state $\rho_{AB_0\cdots B_{N-1}}$, we have
\begin{equation}
\left(E^a_{q,s}\left(\rho_{A|B_0 \cdots B_{N-1}}\right)\right)^{\beta}  \leq
\sum_{j=0}^{N-1}{\beta}^{j}\left(E^a_{q,s}\left(\rho_{A|B_j}\right)\right)^{\beta},
\label{pcrenpoly2}
\end{equation}
conditioned that
\begin{align}
E^a_{q,s}\left(\rho_{A|B_i}\right)\geq \sum_{j=i+1}^{N-1}E^a_{q,s}\left(\rho_{A|B_{j}}\right),
\label{cond3}
\end{align}
for $i=0, \cdots , N-2$.
\label{tpoly2}
\end{Thm}

\begin{proof}
Due to Inequality~(\ref{qbitalppoly1}), it is enough to show
\begin{align}
\left(\sum_{j=0}^{N-1}E^a_{q,s}\left(\rho_{A|B_j}\right)\right)^{\beta}\leq&
\sum_{j=0}^{N-1}{\beta}^{j}\left(E^a_{q,s}\left(\rho_{A|B_j}\right)\right)^{\beta},
\label{qbitalppoly3}
\end{align}
and we use mathematical induction on $N$, Moreover, Inequality~(\ref{p3scpoly3}) assures the validity of Inequality~(\ref{qbitalpmono3})
for $N=2$.

Now, let us assume Inequality~(\ref{qbitalppoly3}) is true for any nonnegative integer less than $N$, and consider
a multi-qubit state $\rho_{AB_0 \cdots B_{N-1}}$. From the equality
\begin{align}
\left(\sum_{j=0}^{N-1}E^a_{q,s}\left(\rho_{A|B_j}\right)\right)^{\beta}
=&\left(E^a_{q,s}\left(\rho_{A|B_0}\right)\right)^{\beta}
\left(1+\frac{\sum_{j=1}^{N-1}E^a_{q,s}\left(\rho_{A|B_j}\right)}
{E^a_{q,s}\left(\rho_{A|B_0}\right)} \right)^{\beta},
\label{pnscpoly5}
\end{align}
and Inequality~(\ref{betle1}) together with the condition in Inequality~(\ref{cond3}), we have
\begin{align}
\left(\sum_{j=0}^{N-1}E^a_{q,s}\left(\rho_{A|B_j}\right)
\right)^{\beta}\leq&
\left(E^a_{q,s}\left(\rho_{A|B_0}\right)\right)^{\beta}
+\beta\left(\sum_{j=1}^{N-1}E^a_{q,s}\left(\rho_{A|B_j}\right)
\right)^{\beta}.
\label{pnscpoly6}
\end{align}

Because the summation of the right-hand side in Inequality~(\ref{pnscpoly6}) is a summation of $N-1$ terms,
thus, the induction hypothesis leads us to
\begin{align}
\left(\sum_{j=1}^{N-1}E^a_{q,s}\left(\rho_{A|B_j}\right)
\right)^{\beta}\leq
\sum_{j=1}^{N-1}{\beta}^{j-1}\left(E^a_{q,s}\left(\rho_{A|B_j}\right)\right)^{\beta}.
\label{pnscpoly7}
\end{align}
Now, Inequalities~(\ref{pnscpoly6}) and (\ref{pnscpoly7}) recover Inequality~(\ref{qbitalppoly3}),
and this completes the proof.
\end{proof}

From Inequality~(\ref{numcom}), we have $\omega_{H}\left(\overrightarrow{j}\right)\leq j$
for any nonnegative integer $j$ and its corresponding binary vector $\overrightarrow{j}$,
therefore
\begin{align}
\left(E^a_{q,s}\left(\rho_{A|B_0 \cdots B_{N-1}}\right)\right)^{\beta}  \leq&
\sum_{j=0}^{N-1}{\beta}^{j}\left(E^a_{q,s}\left(\rho_{A|B_j}\right)\right)^{\beta}\nonumber\\
\leq& \sum_{j=0}^{N-1}{\beta}^{\omega_{H}\left(\overrightarrow{j}\right)}\left(E^a_{q,s}\left(\rho_{A|B_j}\right)\right)^{\beta},
\label{ineqtight3}
\end{align}
for $0\leq \beta \leq 1$. Thus, Inequality~(\ref{pcrenpoly2}) of Theorem~\ref{tpoly2} is tighter than Inequality~(\ref{pcrenpoly0}) of
Theorem~\ref{tpoly0} for $0\leq \beta \leq 1$ and any multi-qubit state $\rho_{AB_0 B_1 \cdots B_{N-1}}$ satisfying the condition in Inequality~(\ref{cond3}).

\section{Conclusions}\label{Sec: Conclusion}
We have provided a characterization of multi-qubit entanglement monogamy and polygamy constraints in terms of non-negative power of entanglement measures based on unified entropy. Using the Hamming weight of the binary vector related with the distribution of subsystems, we have established a class of tight monogamy inequalities of multi-qubit entanglement based on the $\alpha$th-power of UE for $\alpha \geq 1$.
We have further established  a class of tight polygamy inequalities of multi-qubit entanglement in terms of the $\beta$th-power of UEoA for $0 \leq \beta \leq 1$.

Our results presented here deal with the full range of non-negative power of the most general class of bipartite entanglement measures based on unified-$(q,s)$ entropy to establish the monogamy and polygamy inequalities of multi-qubit entanglement so that our results encapsulate the results of R\'enyi and Tsallis entanglement-based multi-qubit entanglement constraints as special cases. Furthermore, our class of monogamy and polygamy inequalities hold in a tighter way, which can also provide finer
characterizations of the entanglement shareability and distribution among the multi-qubit systems. Noting the importance of the study on multi-party quantum entanglement, our result can provide a useful methodology to understand the monogamy and polygamy nature of multi-party quantum entanglement.

\section*{Acknowledgments}
This research was supported by Basic Science Research Program through the National Research Foundation of Korea(NRF)
funded by the Ministry of Education(NRF-2017R1D1A1B03034727).


\end{document}